\documentclass[a4paper,USenglish,cleveref,autoref,thm-restate,nolineno]{socg-lipics-v2021}
\usepackage{microtype}
\usepackage{wrapfig}
\usepackage{graphicx}
\usepackage{xspace}
\usepackage{multicol}
\usepackage{cite}
\usepackage{upgreek}
\usepackage{varwidth}
\usepackage[table]{xcolor}
\usepackage[ruled,vlined,linesnumbered]{algorithm2e}
\graphicspath{{img/}}
\usepackage{blindtext}
\usepackage{paralist}

\newcommand{\RE}{\mathbb{R}}

\newcommand{\bigOh}{\mathcal{O}}

\newcommand{\OO}[1]{O\kern-2pt\left(#1\right)}  

\newcommand{\etal}{\textit{et al.}}

\newcommand{\trainingSet}{\ensuremath{P}\xspace}
\newcommand{\reducedSet}{\ensuremath{R}\xspace}

\newcommand{\ie}{\emph{i.e.},\xspace}
\newcommand{\eg}{\emph{e.g.},\xspace}

\newcommand{\numBorder}{k}

\definecolor{yellowcd}{RGB}{252, 229, 30}
\definecolor{bluecd}{RGB}{51, 0, 68}

\title{Improved Search of Relevant Points for Nearest-Neighbor Classification} 

\author{Alejandro Flores-Velazco}{
    Department of Computer Science\\
    University of Maryland, College Park, MD, USA}
    {afloresv@umd.edu}
    {https://orcid.org/0000-0003-0868-9802}{}
\authorrunning{A.\ Flores\,-Velazco}
\Copyright{A.\ Flores\,-Velazco}

\ccsdesc[500]{Theory of computation~Computational geometry}
\keywords{
    nearest-neighbor classification,
    nearest-neighbor rule,
    decision boundaries,
    border points,
    relevant points}
\category{}
\funding{} 

\EventEditors{}
\EventNoEds{1}
\EventLongTitle{}
\EventShortTitle{}
\EventAcronym{}
\EventYear{2022}
\EventDate{}
\EventLocation{}
\EventLogo{}
\SeriesVolume{}
\ArticleNo{}

\begin{document}

\maketitle

\begin{abstract}
Given a training set $\trainingSet \subset \RE^d$, the \emph{nearest-neighbor classifier} assigns any query point $q \in \RE^d$ to the class of its closest point in \trainingSet. To answer these classification queries, some training points are more relevant than others. We say a training point is \emph{relevant} if its omission from the training set could induce the misclassification of some query point in $\RE^d$. These relevant points are commonly known as \emph{border points}, as they define the boundaries of the Voronoi diagram of \trainingSet that separate points of different classes. Being able to compute this set of points efficiently is crucial to reduce the size of the training set without affecting the accuracy of the nearest-neighbor classifier.

Improving over a decades-long result by~\cite{clarkson1994more}, in a recent paper~\cite{eppstein2022finding} an \emph{output-sensitive} algorithm was proposed to find the set of border points of \trainingSet in $\bigOh( n^2 + n\numBorder^2 )$ time, where $\numBorder$ is the size of such set. In this paper, we improve this algorithm to have time complexity equal to $\bigOh( n\numBorder^2 )$ by proving that the first steps of their algorithm, which require $\bigOh( n^2 )$ time, are unnecessary.
\end{abstract}

\section{Introduction}
\label{sec:intro}

In the context of non-parametric classification, we are given a training set $\trainingSet \subset \RE^d$ consisting of $n$ \emph{labeled} points in $d$-dimensional Euclidean space, where the label of every point in \trainingSet indicates the \emph{class} (or \emph{color}) that the point belongs to. The goal of a classifier is to use the training set \trainingSet to \emph{predict} the class for any \emph{unlabeled} query point $q \in \RE^d$, that is, to \emph{classify} $q$.

The \emph{nearest-neighbor classifier} (also known as \emph{nearest-neighbor rule})~\cite{fix_51_discriminatory} stands out as a simple yet powerful method, that works by assigning any query point $q$ to the class of its closest point in \trainingSet. Despite its simplicity, the nearest-neighbor classifier is well-known to exhibit good classification accuracy both experimentally and theoretically~\cite{stone1977consistent, Cover:2006:NNP:2263261.2267456, devroye1981inequality}. In fact, it is still frequently used in many applications~\cite{boiman2008defense, JDH17, avq_2020, DBLP:journals/corr/abs-1905-01019, peri2020deep, sitawarin2019robustness, papernot2018deep} over more recent and sophisticated techniques like support-vector machines~\cite{Cortes95support-vectornetworks} and deep neural networks~\cite{DBLP:journals/corr/Schmidhuber14}.

One of the principal disadvantages of this technique is its high dependency on the size and dimensionality of the data, especially in light of \emph{big data} applications. With training sets with billions of points becoming increasingly common, reducing the nearest-neighbor classifier's dependency on $n$ and $d$ is one approach to enhance its efficiency. There has been significant progress towards this goal, mainly focusing on two directions. The first involves the design of efficient data structures to answer approximate nearest-neighbor queries~\cite{DBLP:conf/esa/Flores-VelazcoM21, johnson2017billionscale, guo2020accelerating, avd-harpeled, arya2009space, arya2018approximate, arya2017optimal}. The second direction focuses on reducing the size of the training set used by the nearest-neighbor classifier, thus effectively reducing $n$. However, most practical techniques for training set reduction provide limited guarantees on the effect of this reduction to the accuracy of the nearest-neighbor classifier~\cite{DBLP:conf/cccg/Flores-VelazcoM19, cccg20afloresv, angiulli2007fast, floresvelazco_et_al:LIPIcs:2020:12913}.

\begin{figure*}[ht]
    \centering
    \begin{subfigure}[b]{.35\linewidth}
        \centering\includegraphics[width=\textwidth]{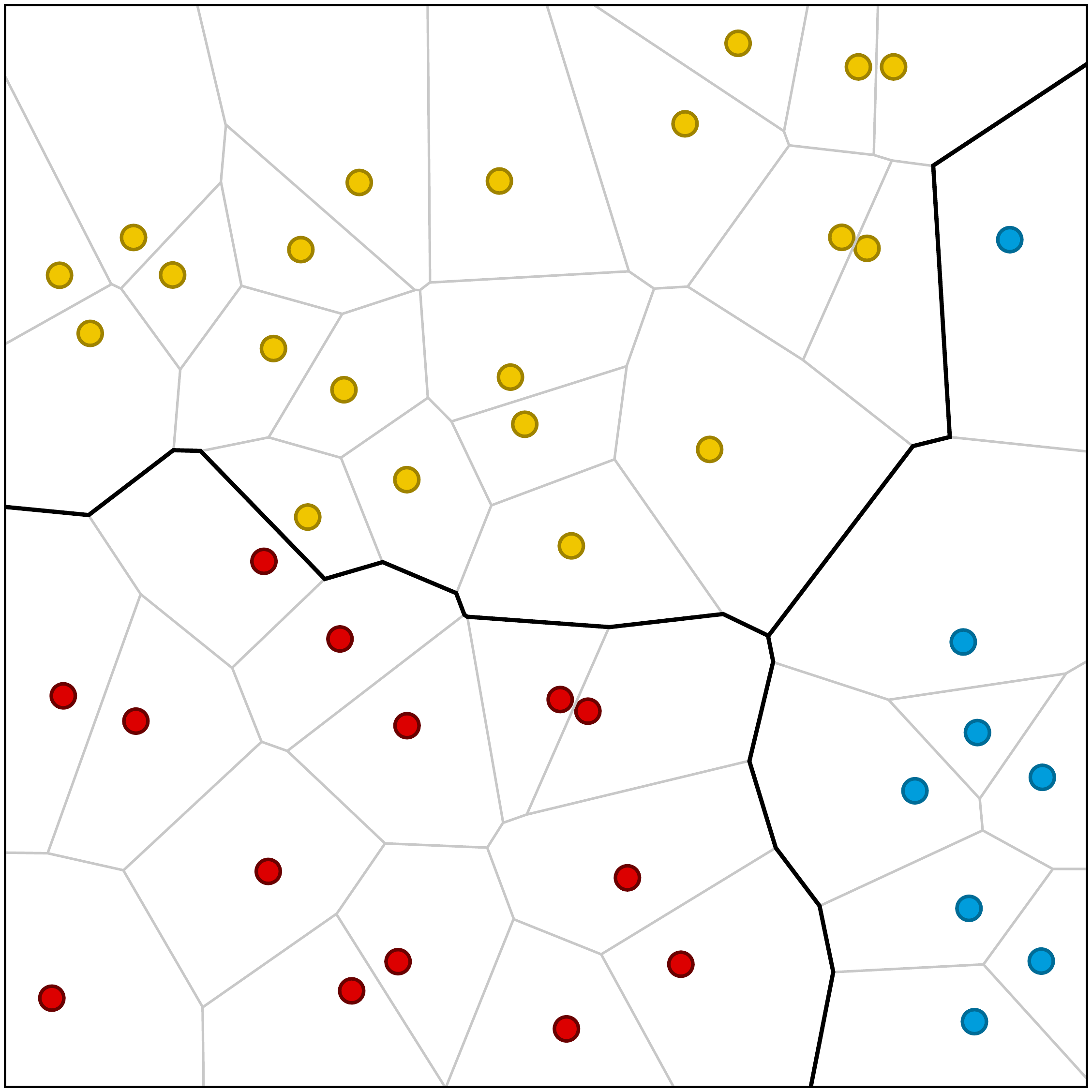}
        \caption{Original training set \trainingSet}
		\label{fig:example:set-training}
    \end{subfigure}%
    \hspace*{0.5cm}
    \begin{subfigure}[b]{.35\linewidth}
        \centering\includegraphics[width=\textwidth]{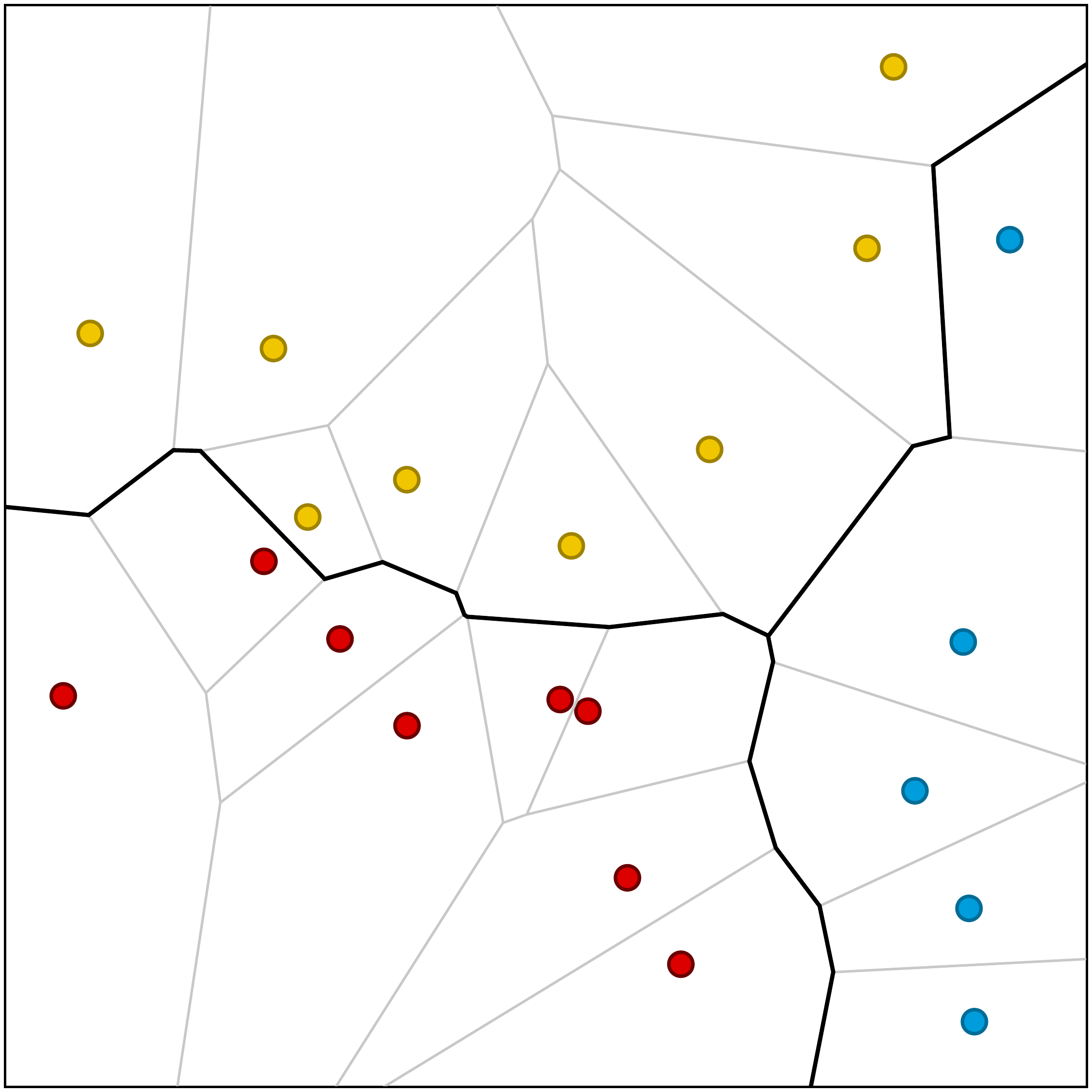}
        \caption{Border/Relevant points of \trainingSet}
		\label{fig:example:set-border}
    \end{subfigure}%
    \caption{On the left, a training set \trainingSet with points of three classes: \emph{red}, \emph{blue} and \emph{yellow}. There the black lines highlight the \emph{boundaries} of \trainingSet between points of different classes. On the right, a subset of these points corresponding to the set of border points of \trainingSet. Note that by definition, the boundaries between points of different classes remain the same for \trainingSet and for its set of border points.}
    \label{fig:example:set}
\end{figure*}

Only a handful of works~\cite{clarkson1994more,Bremner2003,eppstein2022finding} have proposed training set reduction algorithms that guarantee the same classification of every query point, before and after the reduction took place. These are called \emph{boundary preserving} algorithms, and it is the focus of this paper.

The set of \emph{border points} (or \emph{relevant points}\footnote{While~\cite{eppstein2022finding} uses the term \emph{relevant points}, the term \emph{border points} has been the standard in the literature of this and other related problems~\cite{DBLP:conf/cccg/Flores-VelazcoM19, cccg20afloresv, jankowski2004comparison, Wilson:1997:IPT:645526.657143}. For this reason, we stick to the term \emph{border points}.}) of the training set \trainingSet are those that define the boundaries between points of different classes, and whose omission from the training set would imply the misclassification of some query points in $\RE^d$. Formally, two points $p, \hat{p} \in \trainingSet$ are border points of \trainingSet if they belong to different classes, and there exist some point $q \in \RE^d$ such that $q$ is equidistant to both $p$ and $\hat{p}$, and no other point of \trainingSet is closer to $q$ than these two points (\ie the empty ball property of Voronoi Diagrams). See Figure~\ref{fig:example:set} for an example of a training set \trainingSet in $\RE^2$ and its set of border points. Throughout, we let $\numBorder$ denote the total number of border points in the training set. By definition, if instead of building the nearest-neighbor classifier with the entire training set \trainingSet we use the set of border points of \trainingSet, its dependency is reduced from $n$ to $\numBorder$, while still obtaining the same classification for any query point in $\RE^d$. This becomes particularly relevant for applications where $\numBorder \ll n$.

In this paper, we improve a recently proposed algorithm by~\cite{eppstein2022finding} that computes the set of border points of any training set $\trainingSet \subset \RE^d$, where dimension $d$ is assumed to be constant. While the original algorithm computes such set in $\bigOh(n^2 + n\numBorder^2)$ time, where $\numBorder$ is the number of border points of \trainingSet, our new algorithm computes the same set in $\bigOh(n\numBorder^2)$ time.
\subsection{Previous Work}

Other related problems in the realm of training set reduction are NP-hard~\cite{Wilfong:1991:NNP:109648.109673, khodamoradi2018consistent, Zukhba:2010:NPP:1921730.1921735} to solve exactly (\eg those of finding minimum cardinality \emph{consistent} subsets and \emph{selective} subsets). However, the problem of preserving the class boundaries of the nearest-neighbor classifier, or simply, finding the set of border points of \trainingSet, is tractable.

For training sets $\trainingSet \subset \RE^2$ in 2-dimensional Euclidean space, Bremner \etal~\cite{Bremner2003} proposed an output-sensitive algorithm for finding the set of border points of \trainingSet in $\bigOh(n \log{\numBorder})$ worst-case
time. However, how to generalize this algorithm for higher dimensions remained unclear.

Until very recently, the best result for the higher dimensional case was that of Clarkson~\cite{clarkson1994more}. He proposed an algorithm to find the set of border points of $\trainingSet \subset \RE^d$, with bounded $d$, that runs in $\bigOh( \min{(n^3, \numBorder n^2 \log{n})} )$ worst-case time. For almost three decades, this remained the best result for training sets in $\RE^d$. Recently, Eppstein~\cite{eppstein2022finding} proposed a significantly faster algorithm for the $d$-dimensional Euclidean case, which runs in $\bigOh( n^2 + n \numBorder^2 )$ worst-case time.

Eppstein's algorithm is strikingly simple, yet full of interesting ideas (see Algorithm~\ref{alg:all-border-eppstein}). The algorithm works as follows: it begins by selecting an initial set of border points of \trainingSet, one point from every class region. From here, the algorithm uses a series of subroutines which we will group together and denote as the ``\emph{inversion method}'', to find the remaining border points of \trainingSet. Thus, the algorithm can be naturally split into two steps:~the~initialization of \reducedSet with some border points, and the search process for the remaining border points of \trainingSet.

\begin{algorithm}
 \DontPrintSemicolon
 \vspace*{0.1cm}
 \KwIn{Initial training set \trainingSet}
 \KwOut{The set of border points of \trainingSet}
 Let $M$ be the MST of \trainingSet\;
 Initialize \reducedSet with the end points of every bichromatic edge of $M$\;
 \ForEach{$p \in \reducedSet$}{
  Let $c$ be $p$'s class and $\trainingSet_c$ be the points of \trainingSet that belong to class $c$\;
  Let $S_p$ be the inverted points of $\trainingSet \setminus \trainingSet_c$ around $p$\;
  Find all extreme points of $S_p$ and their corresponding original points $E_p$\;
  $\reducedSet \gets \reducedSet \cup E_p$\;
 }
 \KwRet{\reducedSet}
 \vspace*{0.1cm}
 \caption{Recent Eppstein's algorithm~\cite{eppstein2022finding} to find the set of border points of \trainingSet.}
 \label{alg:all-border-eppstein}
\end{algorithm}

The initialization step (lines 1--2 of Algorithm~\ref{alg:all-border-eppstein}) involves finding a subset of border points such that at least one point for every class region is selected. Eppstein observes that this can be achieved by computing the Minimum Spanning Tree (MST) of \trainingSet, identifying the edges of the MST that connect points of different classes (denoted as \emph{bichromatic} edges), and selecting the endpoints of all such edges. This step takes $\bigOh(n^2)$ time, but we will prove that it is not necessary.

The search step (lines 3--6 of Algorithm~\ref{alg:all-border-eppstein}) is in charge of finding every remaining border point of \trainingSet. This step iterates over all selected points, and for each such point $p$, it performs what we call the inversion method. This method identifies a subset of border points of \trainingSet, which are added to \reducedSet. Once the algorithm has done the inversion method on every point of \reducedSet, it terminates with the guarantee of having selected every border point of \trainingSet.

Given any point $p \in \trainingSet$, the inversion method on $p$ is described in lines 4--6 of Algorithm~\ref{alg:all-border-eppstein}. Let $c$ be $p$'s class, and $\trainingSet_c$ be the points of \trainingSet that belong to class $c$, the inversion method on $p$ consists of:
\begin{inparaenum}[(i)]
    \item inverting all points of $\trainingSet \setminus \trainingSet_c$ around a ball centered at $p$ (call the set of these inverted points as $S_p$ and include $p$ itself in the set),
    \item computing the set of extreme points of $S_p$, and finally
    \item returning the set $E_p$ of those points of \trainingSet that correspond to the extreme points of $S_p$ before inversion.
\end{inparaenum}
For a detailed description and proof of correctness of this method, we refer the reader to Eppstein's paper~\cite{eppstein2022finding}. However, for the purposes of this paper we only need a property presented in Lemma~3 of \cite{eppstein2022finding}: the points in $E_p$ reported by the inversion method are the Delaunay neighbors of $p$ with respect to the set $(\trainingSet \setminus \trainingSet_c) \cup \{ p \}$.

Every call of the inversion method takes $\bigOh(n\numBorder)$ time by leveraging well-known output-sensitive algorithms for computing extreme points. Given that this method is called exclusively on every border point of the training set, this yields a total of $\bigOh(n \numBorder^2)$ time to complete the search step of the algorithm. Overall, this implies that Eppstein's algorithm computes the entire set of border points of \trainingSet in $\bigOh(n^2 + n\numBorder^2)$ worst-case time.

\section{Our Approach}

We propose a simple modification to Eppstein's algorithm, which avoids the step of computing the MST of the training set \trainingSet, along with the subsequent selection of bichromatic edges to produce the initial subset of border points.

Instead, we simply start the search process with any arbitrary point of \trainingSet. The rest of the algorithm remains virtually unchanged (see Algorithm~\ref{alg:all-border-mine} for a formal description). We show that this new approach is not only correct, meaning that it only finds border points of \trainingSet, but also complete, as all border points of \trainingSet are eventually found by our algorithm. Additionally, by avoiding the main bottleneck of the original algorithm, our new algorithm computes the same result in $\bigOh(n\numBorder^2)$ time, eliminating the $\bigOh(n^2)$ term.

\begin{algorithm}
 \DontPrintSemicolon
 \vspace*{0.1cm}
 \KwIn{Initial training set \trainingSet}
 \KwOut{The set of border points of \trainingSet}
 Let $s$ be any ``\emph{seed}'' point from \trainingSet\;
 $\reducedSet \gets \phi$\;
 \ForEach{$p \in \reducedSet \cup \{ s \}$}{
  Let $c$ be $p$'s class and $\trainingSet_c$ be the points of \trainingSet that belong to class $c$\;
  Let $S_p$ be the inverted points of $\trainingSet \setminus \trainingSet_c$ around $p$\;
  Find all extreme points of $S_p$ and their corresponding original points $E_p$\;
  $\reducedSet \gets \reducedSet \cup E_p$\;
 }
 \KwRet{\reducedSet}
 \vspace*{0.1cm}
 \caption{New algorithm to find the set of border points of \trainingSet.}
 \label{alg:all-border-mine}
\end{algorithm}

Before proceeding, it is useful to explore why Eppstein's algorithm computes the MST of the training set \trainingSet. First, note that the original algorithm only applies the inversion method on border points of \trainingSet. In fact, Eppstein's correctness proof relies on it: Lemma~6 in~\cite{eppstein2022finding} proves that all points in $E_p$ are border points by assuming that point $p$ is also a border point. From the description of our algorithm, note that we initially apply the inversion method on a ``seed'' point $s$, which might not be a border point. Therefore, we need to generalize Lemma~6 in~\cite{eppstein2022finding} for the case where $p$ is not a border point of \trainingSet. Additionally, using the points from all bichromatic pairs of the MST of \trainingSet guarantees that Eppstein's algorithm starts the search step with at least one point from every boundary of \trainingSet. Eppstein's completeness proof shows that the search step can then ``move along'' any given boundary and eventually select all its defining points. We show that the search process is far more powerful, and can even ``jump'' between nearby boundaries, thus rendering the MST computation unnecessary.

The following description outlines the necessary steps to prove both the correctness and completeness of our new algorithm, which are unfolded in the rest of this section.
\begin{itemize}
    \item By applying the inversion method to any point of \trainingSet, not necessarily a border point, all reported points are border points of \trainingSet. This is established in Lemma~\ref{lemma:seed}, generalizing the statement of Lemma~6 of~\cite{eppstein2022finding} for non-border points.
    \item For any class boundary of \trainingSet, once the algorithm selects a point from this boundary, it will eventually select every other point defining the same boundary. This is originally proved in Lemma~10~\cite{eppstein2022finding}, however, we provide simpler proofs in Lemmas~\ref{lemma:walls} and~\ref{lemma:boundary}.
    \item 
    Given two disconnected boundaries separated by a class region, we prove that if our algorithm selects a defining point from one of the boundaries, it will eventually select all defining points from both boundaries. This is proved in Lemma~\ref{lemma:boundary-jump}.
\end{itemize}

All together, these lemmas are used to prove the main result: the correctness, completeness, and worst-case time complexity of Algorithm~\ref{alg:all-border-mine}, as stated in Theorems~\ref{thm:main} and~\ref{thm:main:best}.

\begin{figure*}[ht]
    \centering
    \begin{subfigure}[b]{.32\linewidth}
        \centering\includegraphics[width=\textwidth]{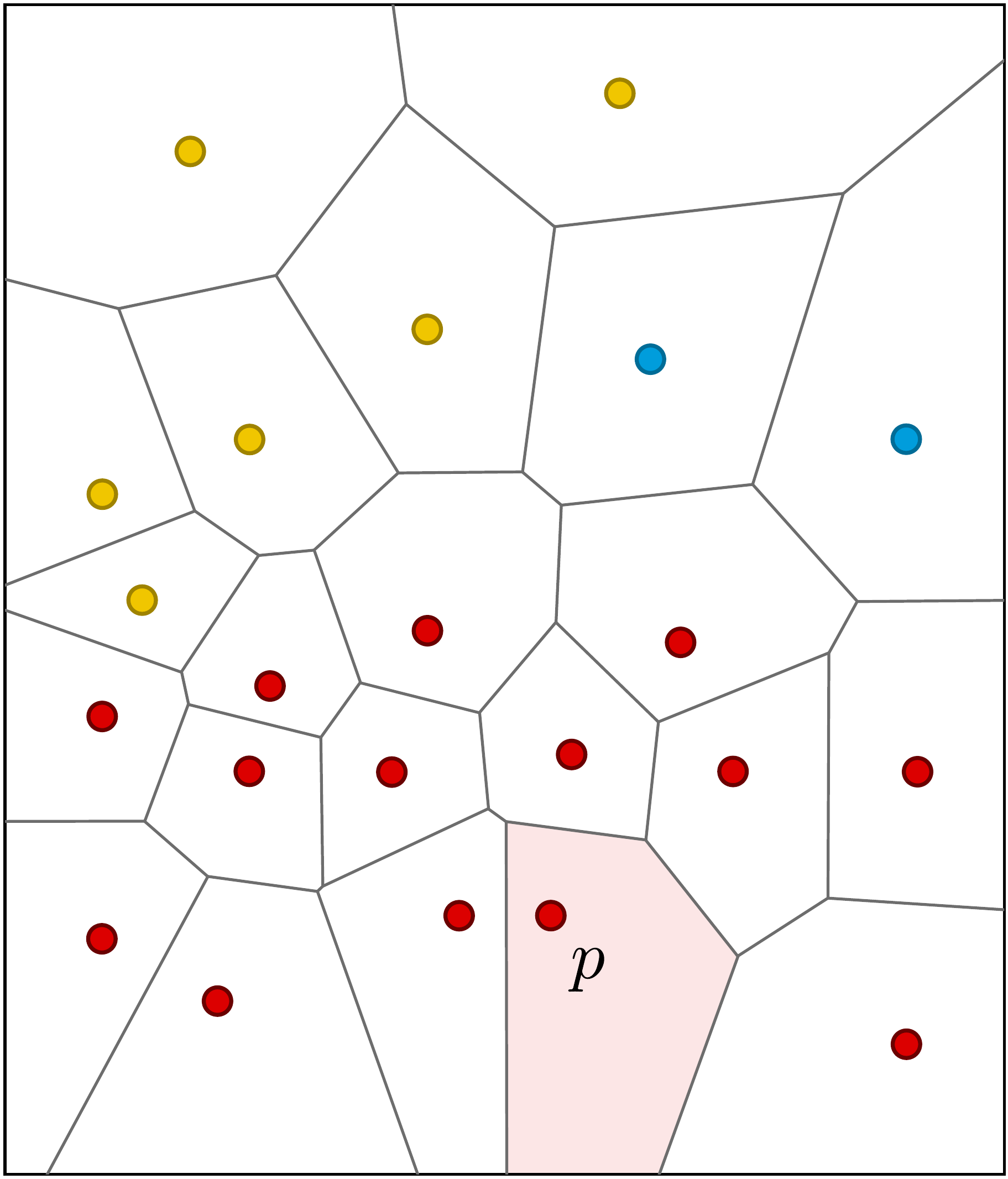}
        \caption{Training set \trainingSet}
		\label{fig:proof:seed-all}
    \end{subfigure}%
    \hfill
    \begin{subfigure}[b]{.32\linewidth}
        \centering\includegraphics[width=\textwidth]{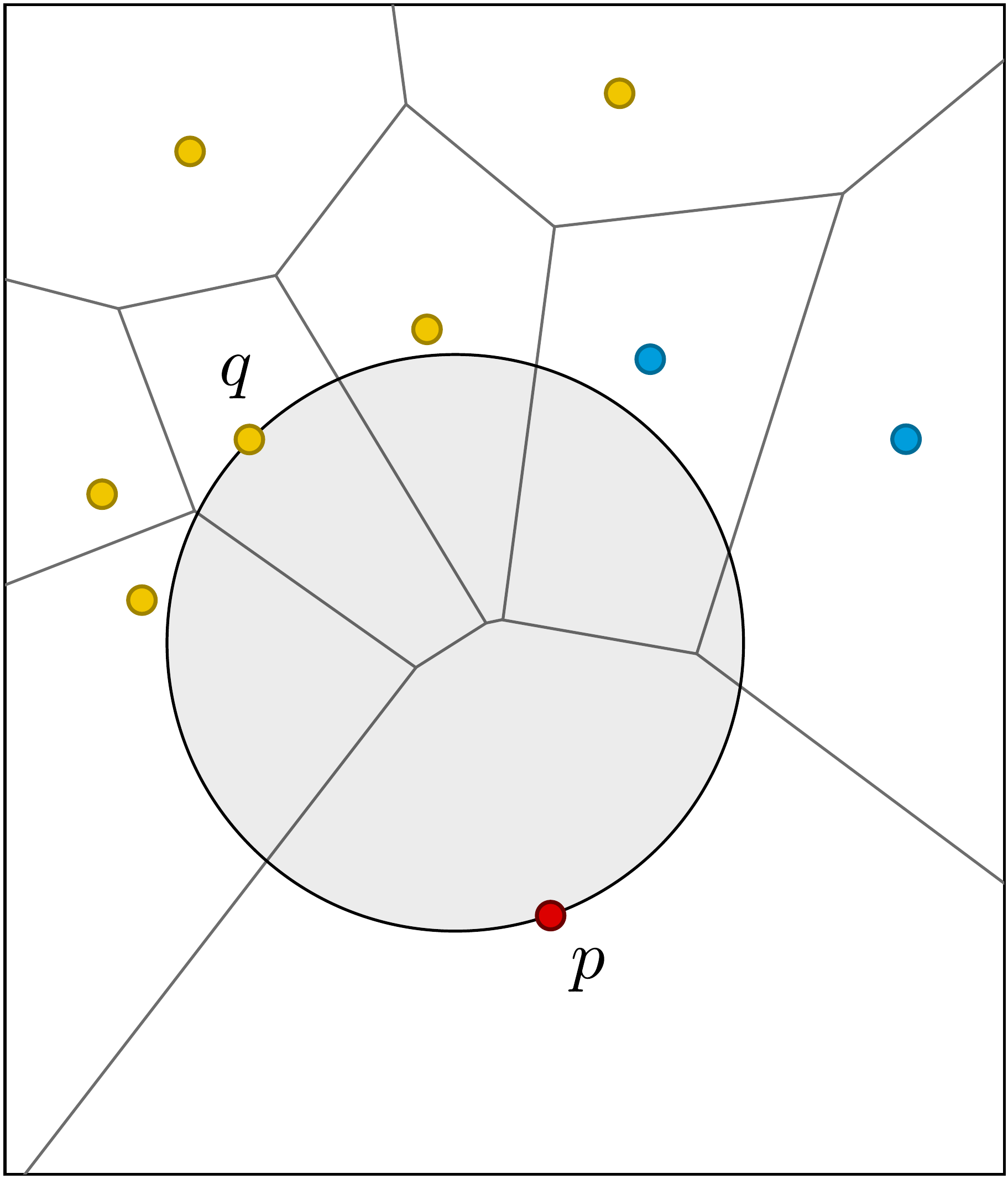}
        \caption{Points from $( \trainingSet \setminus \trainingSet_c ) \cap \{ p \}$}
		\label{fig:proof:seed-inversion}
    \end{subfigure}%
    \hfill
    \begin{subfigure}[b]{.32\linewidth}
        \centering\includegraphics[width=\textwidth]{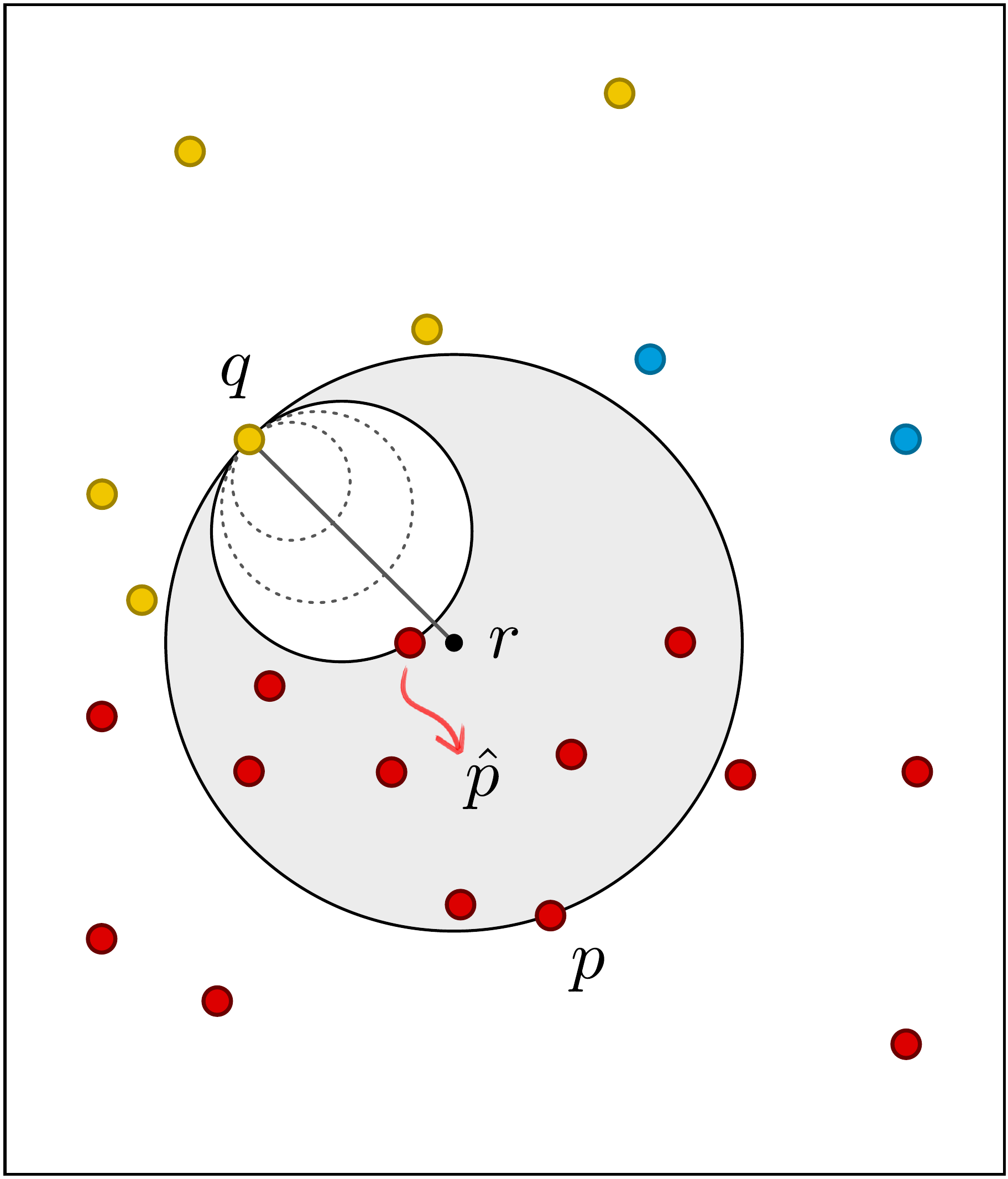}
        \caption{$q$ and $\hat{p}$ are border points}
		\label{fig:proof:seed-witness}
    \end{subfigure}%
    \caption{Example showing the inversion method from any point $p \in \trainingSet$. On the left, training set \trainingSet. The middle figure shows every non red point of \trainingSet, except for $p$ itself, along with a point $q$ selected from the inversion method on $p$. On the right, we see evidence that $q$ is a border point of \trainingSet.}
    \label{fig:proof:seed}
\end{figure*}

\begin{lemma}
\label{lemma:seed}
Let $p \in \trainingSet$ be any point of the training set. Then every point selected using the inversion method on $p$ must be a border point of \trainingSet.
\end{lemma}

\begin{proof}
Let $E_p$ be the points of \trainingSet corresponding (before inversion) to the extreme points of $S_p$. According to Lemma~3~\cite{eppstein2022finding}, every point in $E_p$ is a neighbor of point $p$ with respect to the Voronoi Diagram of set $(\trainingSet \setminus \trainingSet_c) \cup \{ p \}$.
This implies that for every point $q \in E_p$ other than p, there exists a ball such that both $p$ and $q$ are on its surface and no points of $\trainingSet \setminus \trainingSet_c$ lie inside (see Figures~\ref{fig:proof:seed-all} and~\ref{fig:proof:seed-inversion}). We can now leverage similar techniques to the ones described in~\cite{Bremner2003}, to find a ``witness'' point to the hypothesis that $q$ must be a border point of \trainingSet.

Recall that the empty ball we just described, as illustrated in Figure~\ref{fig:proof:seed-inversion}, is empty from points of $\trainingSet \setminus \trainingSet_c$. However, there might be points of $P_c$ inside. And moreover, we know that at least one point of $P_c$, point $p$, lies on its surface. Now, let $r$ be the center of this ball, we grow an empty ball, this time with respect to the entire training set \trainingSet, such that its center lies on the line $\overline{qr}$ and point $q$ is on its surface (see Figure~\ref{fig:proof:seed-witness}). This ball will grow until it hits another point $\hat{p}$ of \trainingSet, which we are guaranteed it will be of the same class as point $p$, and thus, of different class as point $q$. Finally, we have just found an empty ball with respect to \trainingSet, which has points $q$ and $\hat{p}$ on its surface, and were the class of both points differ. Therefore, this implies that $q$ is a border point of \trainingSet.
\end{proof}

\begin{figure*}[ht]
    \centering
    \begin{subfigure}[b]{.32\linewidth}
        \centering\includegraphics[width=\textwidth]{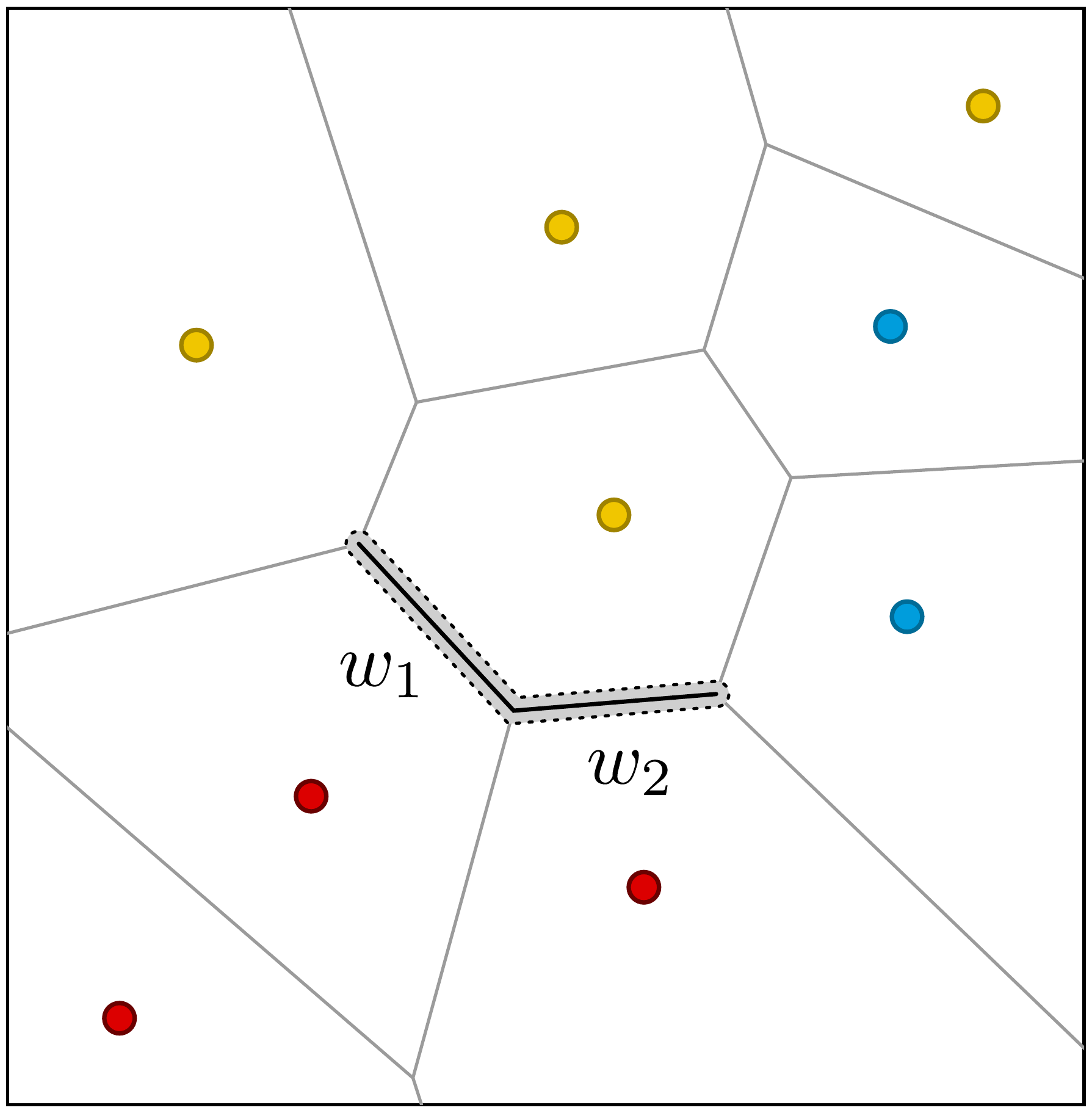}
        \caption{}
		\label{fig:proof:walls-adjacent}
    \end{subfigure}%
    \hspace*{0.2cm}
    \begin{subfigure}[b]{.32\linewidth}
        \centering\includegraphics[width=\textwidth]{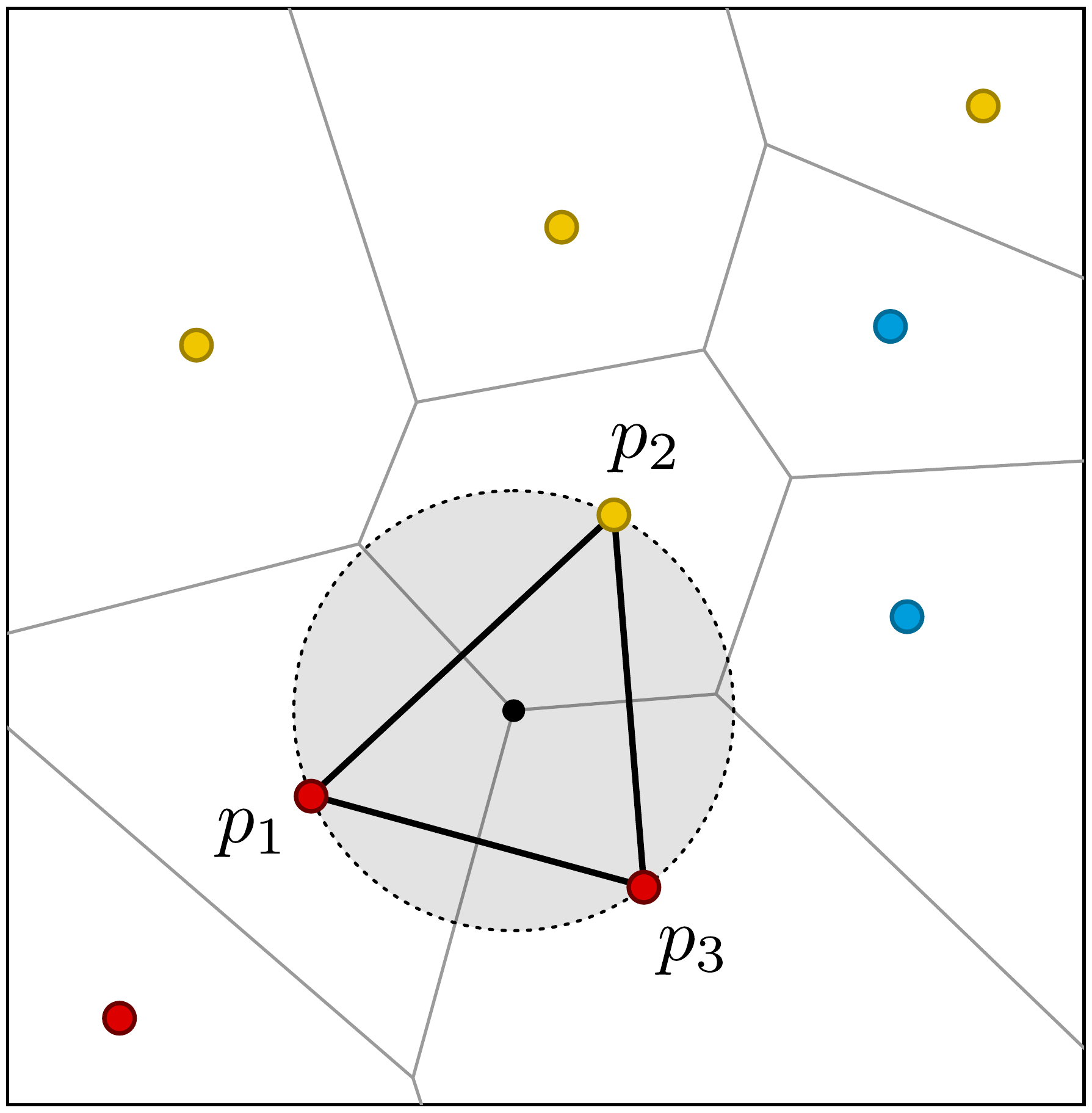}
        \caption{}
		\label{fig:proof:walls-triangle}
    \end{subfigure}%
    \caption{By definition, any two adjacent walls $w_1$ and $w_2$ of the Voronoi Diagram of \trainingSet hold the empty ball property with the points that define them. When these walls are part of the class boundaries of \trainingSet, the points that define them belong to at least two classes.}
    \label{fig:proof:walls}
\end{figure*}

Before continuing, we need to formally define a few concepts. First, we define a \emph{wall} of \trainingSet as any $(d-1)$-dimensional face of the Voronoi Diagram of \trainingSet. By known properties of these structures, every wall $w$ is \emph{defined} by two distinct points $p,q \in \trainingSet$ such that any point on $w$ has $p$ and $q$ as its two equidistant nearest-neighbors in the training set. We say two walls are \emph{adjacent} if their intersection is not empty. That is, if there exists a point in $\RE^d$ with all the defining points of these two walls as its equidistant nearest-neighbors in \trainingSet.

Additionally, we define a \emph{class boundary} (or just \emph{boundary}) of \trainingSet as the union of adjacent walls, where each of these walls is defined by two points of different classes. Similarly, we define a \emph{class region} of \trainingSet as the union of adjacent Voronoi cells whose defining points belong to the same class. Based on these definitions, note that class boundaries are the ones that separate different class regions of \trainingSet. Figure~\ref{fig:example:panoramic-set} illustrates a training set in $\RE^2$ with points of three classes, whose Voronoi Diagram describes five class regions and two class boundaries.

\begin{lemma}
\label{lemma:walls}
Let $w_1$ and $w_2$ be two adjacent walls in a class boundary of \trainingSet. If the algorithm selects one of the points defining one of these walls, it eventually selects the remaining points defining both walls.
\end{lemma}

\begin{proof}
Let $\mathcal{W}$ be the set of points defining both walls $w_1$ and $w_2$ (see Figure~\ref{fig:proof:walls}). By definition, these two walls of the Voronoi Diagram of \trainingSet are adjacent if there exists an empty ball with all the points of $\mathcal{W}$ on its surface. Knowing these two walls are part of the class boundaries of \trainingSet, the set $\mathcal{W}$ must contain at least three points, and at least two classes.

Let $p_1$ be the first point of $\mathcal{W}$ to be selected by the algorithm. When doing the inversion method on point $p_1$, the algorithm will select all points of $\mathcal{W}$ of different class than $p_1$, of which we know there is at least one. Let $p_2$ be one such point. Finally, when doing the inversion method on point $p_2$, the algorithm will select the remaining points of $\mathcal{W}$ of the same class as $p_1$. Therefore, all points of $\mathcal{W}$ will eventually be selected by the algorithm.
\end{proof}



\begin{figure}[ht]
    \centering\includegraphics[width=\textwidth]{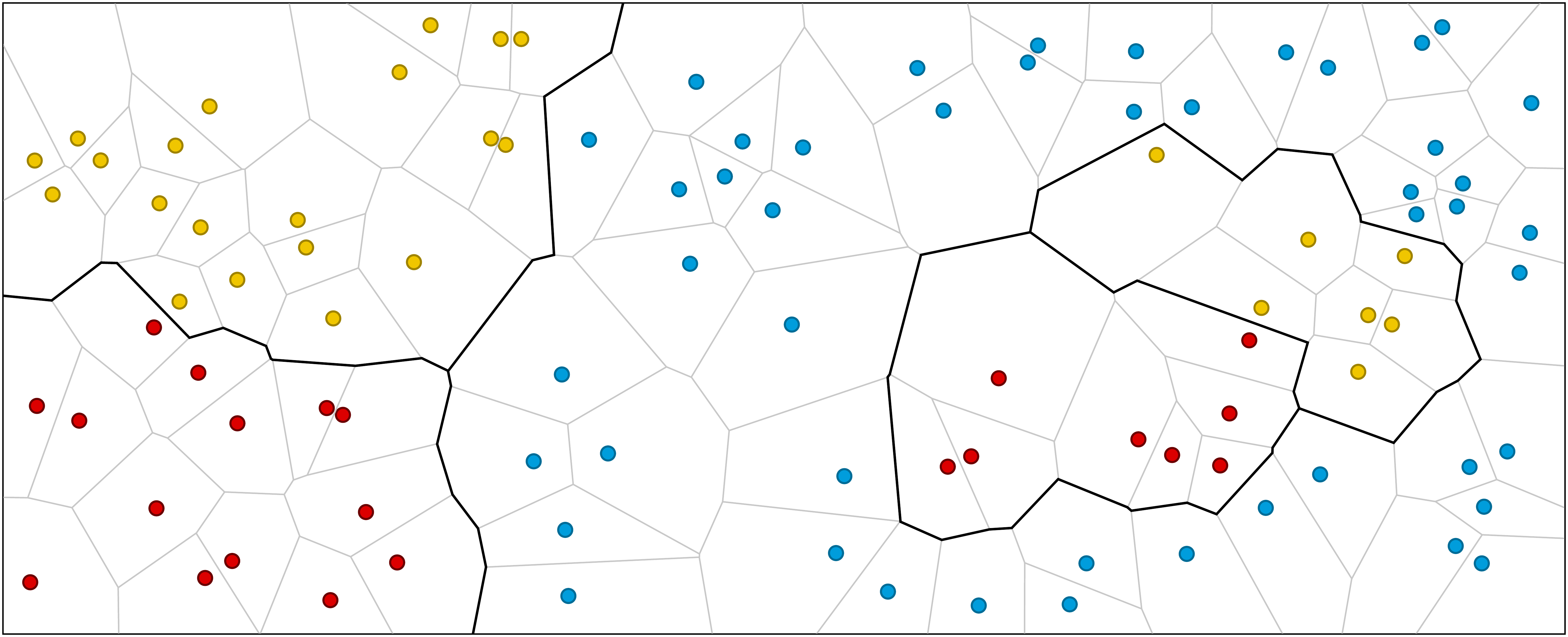}
    \caption{A training set with five class regions (one \emph{blue}, two \emph{red}, and two \emph{yellow} regions), along with two disconnected class boundaries that separate all these regions. On the left, a boundary separating the blue region and the leftmost yellow and red regions. On the right, a boundary that separates the same blue region and the remaining red and yellow regions.}
	\label{fig:example:panoramic-set}
\end{figure}

\begin{lemma}
\label{lemma:boundary}
Let $\mathcal{A}$ be a class boundary of \trainingSet, and assume that the algorithm selects one of the defining points of $\mathcal{A}$. Then, the algorithm will eventually select all defining points of $\mathcal{A}$.
\end{lemma}

This comes as a direct consequence of Lemma~\ref{lemma:walls} and the definition of a class boundary of the training set \trainingSet. It remains to show what happens with boundaries that are disconnected.

\begin{lemma}
\label{lemma:boundary-jump}
Let $\mathcal{A}$ and $\mathcal{B}$ be two disconnected boundaries of \trainingSet, such that there exists a path in space from $\mathcal{A}$ to $\mathcal{B}$ that is completely contained within one color region. Without loss of generality, say that every point that defines $\mathcal{A}$ has been selected by the algorithm. Then, every point that defines $\mathcal{B}$ must also be selected by the algorithm.
\end{lemma}

\begin{proof}
Given these two disconnected boundaries $\mathcal{A}$ and $\mathcal{B}$, we assume there exists some path $\mathcal{P}$ in $\RE^d$ going from a wall of $\mathcal{A}$ to a wall of $\mathcal{B}$, such that this path passes exclusively through a single class region (see Figure~\ref{fig:proof:boundaries-regions}). Without loss of generality, say this is a \emph{red} class region. Formally, for every point $r$ along $\mathcal{P}$ we know $r$'s nearest-neighbor in \trainingSet is red. Additionally, we assume that every border point defining $\mathcal{A}$ is selected by the algorithm. Hence, the proof consists of showing that there exists a sequence of border points $\langle p_1, \hat{p}_1, p_2, \hat{p}_2, \dots, p_m, \hat{p}_m \rangle$ such that
\begin{inparaenum}[(i)]
    \item $p_1$ and $\hat{p}_m$ are defining points of $\mathcal{A}$ and $\mathcal{B}$, respectively,
    \item $\hat{p}_i$ is retrieved by the inversion method on $p_i$, for every $i \in [1,m]$, and finally
    \item points $p_i$ and $\hat{p}_{i-1}$ are both defining the same boundary, for every $i \in [2,m]$.
\end{inparaenum}
See Figure~\ref{fig:proof:boundaries} for a visual description.

\begin{figure*}[ht]
    \centering
    \begin{subfigure}[b]{.32\linewidth}
        \centering\includegraphics[width=\textwidth]{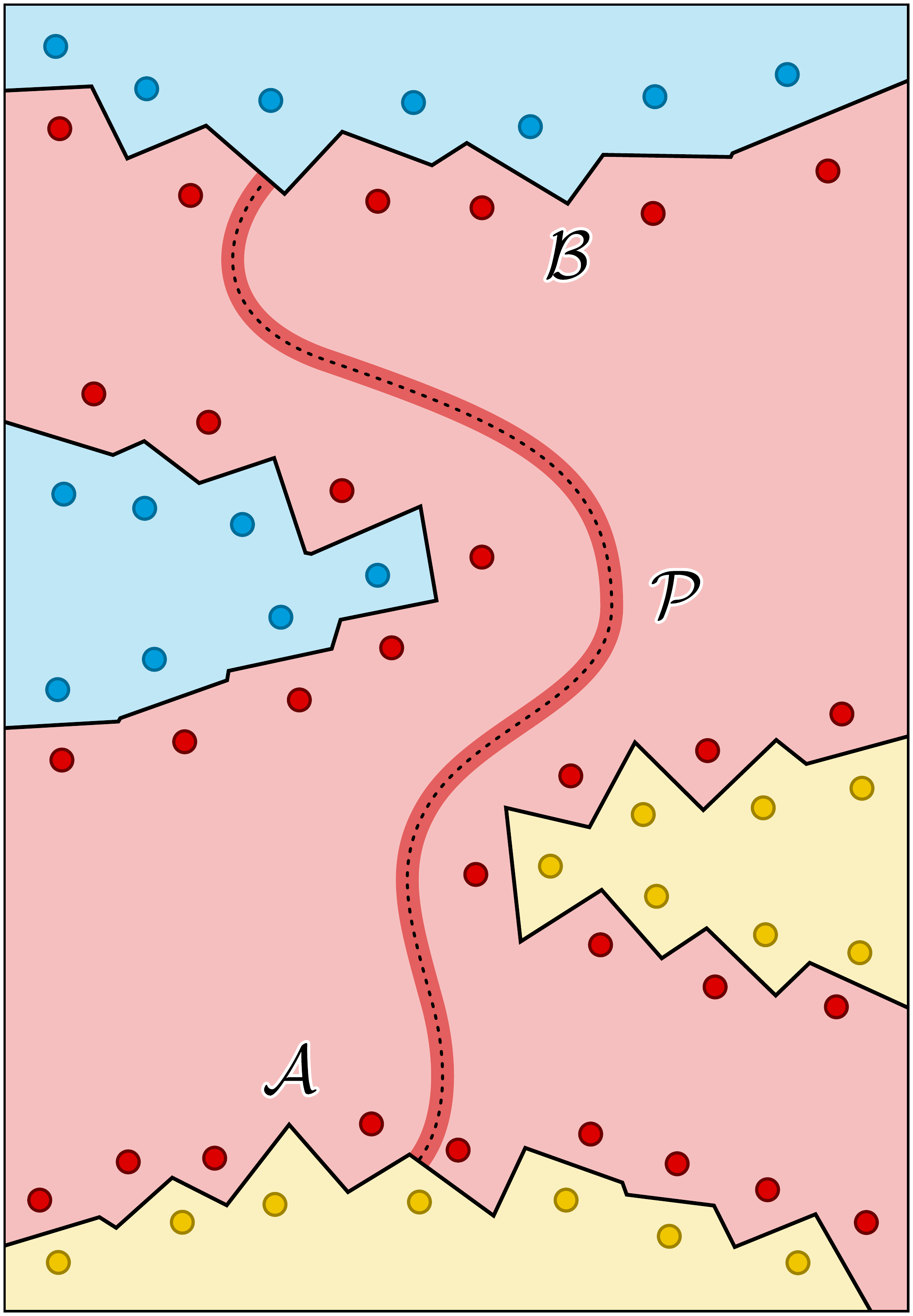}
        \caption{}
		\label{fig:proof:boundaries-regions}
    \end{subfigure}%
    \hfill
    \begin{subfigure}[b]{.32\linewidth}
        \centering\includegraphics[width=\textwidth]{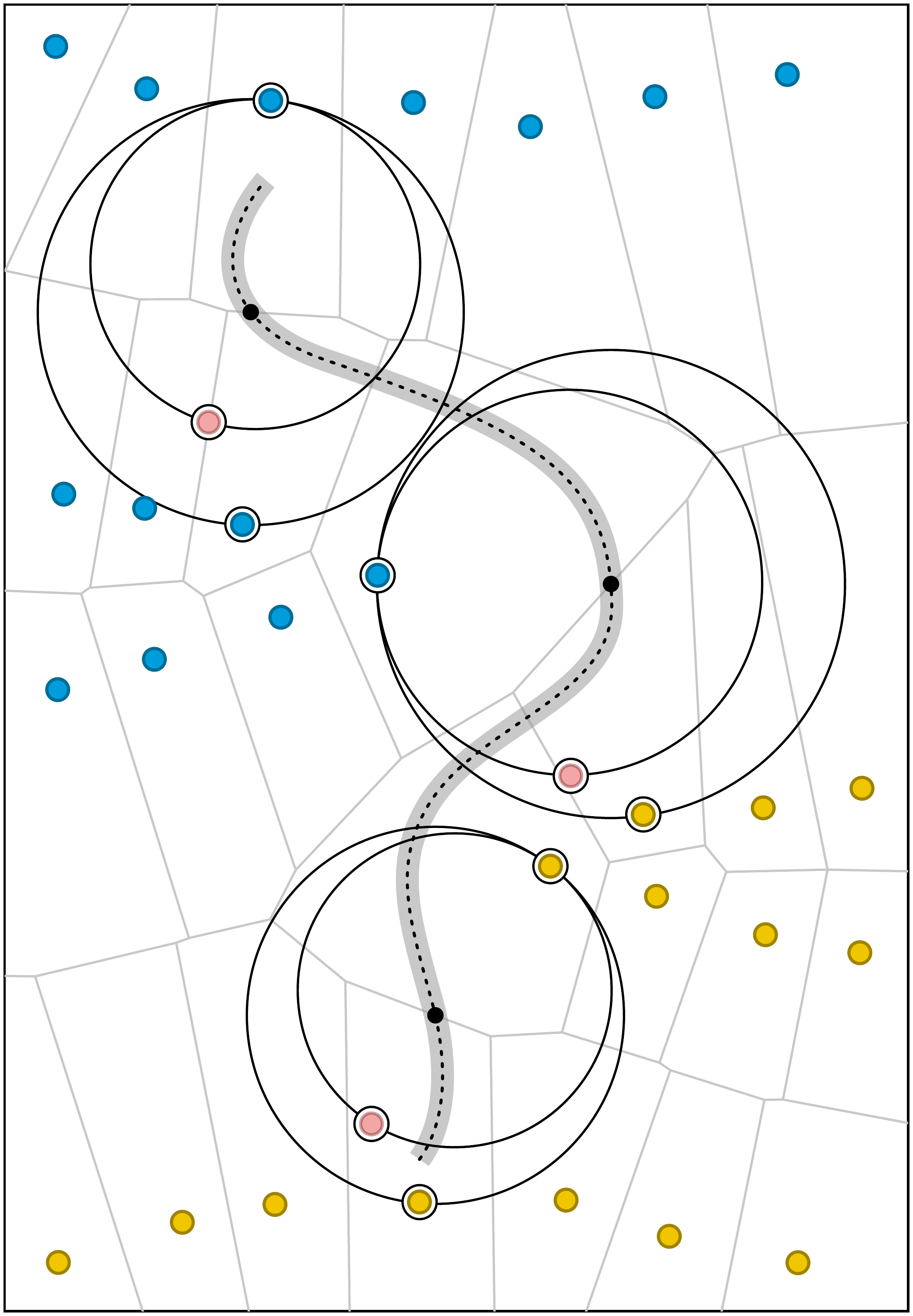}
        \caption{}
		\label{fig:proof:boundaries-path}
    \end{subfigure}%
    \hfill
    \begin{subfigure}[b]{.32\linewidth}
        \centering\includegraphics[width=\textwidth]{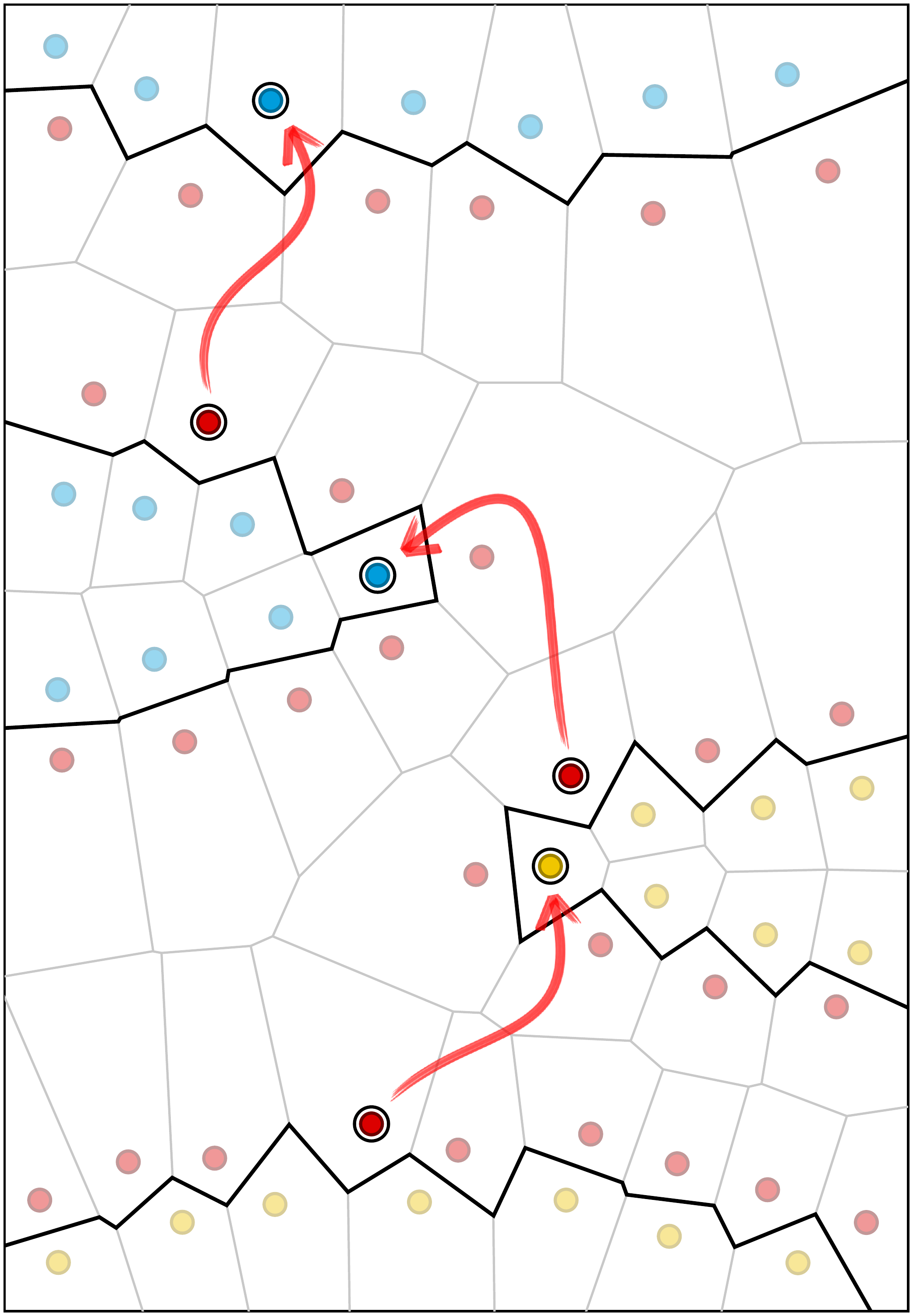}
        \caption{}
		\label{fig:proof:boundaries-jump}
    \end{subfigure}%
    \caption{On the right, two disconnected boundaries $\mathcal{A}$ and $\mathcal{B}$ enclosing a red class region. Thus, there is a path $\mathcal{P}$ completely contained inside such region and connecting both boundaries. Other boundaries can also be enclosing the same region and be near path $\mathcal{P}$. On the left, we proof that there exists a sequence of points that can be retrieved by calls to the inversion method, such that if points of $\mathcal{A}$ are selected by the algorithm, eventually points of $\mathcal{B}$ will also be selected.}
    \label{fig:proof:boundaries}
\end{figure*}


By definition, for every point $r$ along path $\mathcal{P}$ we know $r$'s nearest-neighbor is a red point. Now, let's delete every red point from consideration, including the ones defining boundaries $\mathcal{A}$ and $\mathcal{B}$ (see Figure~\ref{fig:proof:boundaries-path}). This immediately implies that $r$'s nearest-neighbor just became a non-red border point of \trainingSet. The fact that $r$'s new nearest-neighbor is a border point is easy to proof, using similar arguments as the ones laid down in Lemma~\ref{lemma:seed}. Additionally, these border points could be defining other boundaries apart from $\mathcal{A}$ and $\mathcal{B}$, as seen in Figure~\ref{fig:proof:boundaries-path}.

Let's start moving along the path $\mathcal{P}$, starting from the end-point of the path that lies on a wall of boundary $\mathcal{A}$. Then, find all $r_i$ points along the path, where each $r_i$ has two equidistant nearest-neighbors among the remaining non-red points, and both points define two distinct boundaries of \trainingSet. We say there are $m$ of these points along the path, and denote $r_i$'s two equidistant nearest-neighbors as $q_{i,1}$ and $q_{i,2}$ for $i \in [1,m]$. Clearly, $q_{i,1}$ and $q_{i-1,2}$ are border points defining the same boundary, for all $i \in [2,m]$. See Figure~\ref{fig:proof:boundaries-path}, where the three black points along the path are the $r_i$ points, and the \emph{yellow} and \emph{blue} points on the surface of the balls centered at each $r_i$ are the corresponding $q_{i,1}$ and $q_{i,2}$ points.

For now, let's fix the analysis on one such $r_i$ point, and consider the ball centered at $r_i$ with both $q_{i,1}$ and $q_{i,2}$ on its surface. There must exist some other point $q_{i,3}$ lying inside of $r_i$'s ball, such that $q_{i,3}$ is one of the deleted red points defining the same boundary as $q_{i,1}$. 
It is now easy to see that there exist an empty ball, with respect to the set $\trainingSet \setminus \trainingSet_\emph{red} \cup \{ q_{i,3} \}$, with both $q_{i,3}$ and $q_{i,2}$ on its boundary. This implies that $q_{i,2}$ is retrieved by the inversion method on $q_{i,3}$.
Therefore, let's add $p_i \gets q_{i,3}$ and $\hat{p}_i \gets q_{i,2}$ to the sequence of points that we are looking for. Repeat this for every $r_i$ with $i \in [1,m]$ to identify all points in the sequence.

Finally, we have the sequence of border points $\langle p_1, \hat{p}_1, p_2, \hat{p}_2, \dots, p_m, \hat{p}_m \rangle$ such that for any $i \in [1,m]$ assuming that the algorithm selects the points defining the same boundary as $p_i$, it will also select $\hat{p}_i$, and leveraging Lemma~\ref{lemma:boundary} it will eventually select all other points defining the same boundary as $\hat{p}_i$. Given that $p_1$ and $\hat{p}_m$ are defining border points of boundaries $\mathcal{A}$ and $\mathcal{B}$, respectively, and by the assumption that all points defining $\mathcal{A}$ are selected by the algorithm, we know that eventually, all points defining $\mathcal{B}$ will be selected too.
\end{proof}

\begin{theorem}
\label{thm:main}
The algorithm selects every border point of \trainingSet in $\bigOh( n \numBorder^2 )$ time.
\end{theorem}

\begin{proof}
Proving the worst-case time complexity of our algorithm follows directly from the time complexity of the search step of Eppstein's algorithm~\cite{eppstein2022finding}. However, the correctness and completeness of our algorithm follows from Lemmas~\ref{lemma:seed} to~\ref{lemma:boundary-jump}.

First, we know by Lemmas~\ref{lemma:seed}-\ref{lemma:boundary} that Algorithm~\ref{alg:all-border-mine} will select the defining border points of at least one class boundary of \trainingSet. Denote this boundary as $\mathcal{A}$ and consider any other boundary $\mathcal{B}$ of \trainingSet. Evidently, we can draw a path $\mathcal{P}$ from $\mathcal{A}$ to $\mathcal{B}$, which would generally pass through several class regions. Then, let's split $\mathcal{P}$ into several subpaths $\mathcal{P}_1, \mathcal{P}_2, \dots, \mathcal{P}_m$ such that each subpath is completely contained within a single class region. From this, we can directly apply Lemma~\ref{lemma:boundary-jump} on each of the intermediate boundaries that ``cut'' $\mathcal{P}$ into these subpaths. Finally, this implies that our algorithm will eventually select every defining point of boundary $\mathcal{B}$, and similarly, it will do the same with all other boundaries of \trainingSet.
\end{proof}

\begin{theorem}
\label{thm:main:best}
Leveraging Chan's algorithm~\cite{chan1996output} for finding extreme points, the algorithm selects every border point of \trainingSet in randomized expected time $\bigOh( n\numBorder \log{\numBorder} )$ for $d=3$, and in
\[
\bigOh \left( \numBorder (n\numBorder)^{1-\frac{1}{\lfloor d/2 \rfloor +1}} (\log{n})^{\bigOh(1)} \right)
\]
time for all constant dimensions $d > 3$.
\end{theorem}

Just as with Eppstein's original algorithm, we can use Chan's randomized algorithm~\cite{chan1996output} for finding extreme points of point sets in $\RE^d$, in order to reduce the expected time complexity of our improved algorithm. The remaining of the proof is the same as for Theorem~\ref{thm:main}.

\paragraph*{Acknowledgements}
Thanks to Prof. David Mount for pointing out Eppstein's paper~\cite{eppstein2022finding} and for the valuable discussions on the results presented in this paper. 

\bibliographystyle{plainurl}
\bibliography{nnc}

\begin{thebibliography}{10}

\bibitem{angiulli2007fast}
Fabrizio Angiulli.
\newblock Fast nearest neighbor condensation for large data sets
  classification.
\newblock {\em IEEE Transactions on Knowledge and Data Engineering},
  19(11):1450--1464, 2007.

\bibitem{arya2017optimal}
Sunil Arya, Guilherme~D. da~Fonseca, and David~M. Mount.
\newblock Optimal approximate polytope membership.
\newblock In {\em Proceedings of the Twenty-Eighth Annual ACM-SIAM Symposium on
  Discrete Algorithms}, pages 270--288. SIAM, 2017.

\bibitem{arya2018approximate}
Sunil Arya, Guilherme~D. Da~Fonseca, and David~M. Mount.
\newblock Approximate polytope membership queries.
\newblock {\em SIAM Journal on Computing}, 47(1):1--51, 2018.

\bibitem{arya2009space}
Sunil Arya, Theocharis Malamatos, and David~M. Mount.
\newblock Space-time tradeoffs for approximate nearest neighbor searching.
\newblock {\em Journal of the ACM (JACM)}, 57(1):1, 2009.

\bibitem{boiman2008defense}
Oren Boiman, Eli Shechtman, and Michal Irani.
\newblock In defense of nearest-neighbor based image classification.
\newblock In {\em 2008 IEEE Conference on Computer Vision and Pattern
  Recognition}, pages 1--8. IEEE, 2008.

\bibitem{Bremner2003}
David Bremner, Erik Demaine, Jeff Erickson, John Iacono, Stefan Langerman, Pat
  Morin, and Godfried Toussaint.
\newblock Output-sensitive algorithms for computing nearest-neighbour decision
  boundaries.
\newblock In Frank Dehne, J{\"o}rg-R{\"u}diger Sack, and Michiel Smid, editors,
  {\em Algorithms and Data Structures: 8th International Workshop, WADS 2003,
  Ottawa, Ontario, Canada, July 30 - August 1, 2003. Proceedings}, pages
  451--461, Berlin, Heidelberg, 2003. Springer Berlin Heidelberg.
\newblock URL: \url{https://doi.org/10.1007/978-3-540-45078-8_39}, \href
  {http://dx.doi.org/10.1007/978-3-540-45078-8_39}
  {\path{doi:10.1007/978-3-540-45078-8_39}}.

\bibitem{chan1996output}
Timothy~M Chan.
\newblock Output-sensitive results on convex hulls, extreme points, and related
  problems.
\newblock {\em Discrete \& Computational Geometry}, 16(4):369--387, 1996.

\bibitem{clarkson1994more}
Kenneth~L Clarkson.
\newblock More output-sensitive geometric algorithms.
\newblock In {\em Proceedings 35th Annual Symposium on Foundations of Computer
  Science}, pages 695--702. IEEE, 1994.

\bibitem{Cortes95support-vectornetworks}
Corinna Cortes and Vladimir Vapnik.
\newblock Support-vector networks.
\newblock In {\em Machine Learning}, pages 273--297, 1995.

\bibitem{Cover:2006:NNP:2263261.2267456}
T.~Cover and P.~Hart.
\newblock Nearest neighbor pattern classification.
\newblock {\em IEEE Trans. Inf. Theor.}, 13(1):21--27, January 1967.
\newblock URL: \url{http://dx.doi.org/10.1109/TIT.1967.1053964}, \href
  {http://dx.doi.org/10.1109/TIT.1967.1053964}
  {\path{doi:10.1109/TIT.1967.1053964}}.

\bibitem{devroye1981inequality}
Luc Devroye.
\newblock On the inequality of cover and hart in nearest neighbor
  discrimination.
\newblock {\em Pattern Analysis and Machine Intelligence, IEEE Transactions
  on}, pages 75--78, 1981.

\bibitem{eppstein2022finding}
David Eppstein.
\newblock Finding relevant points for nearest-neighbor classification.
\newblock In {\em Symposium on Simplicity in Algorithms (SOSA)}, pages 68--78.
  SIAM, 2022.

\bibitem{fix_51_discriminatory}
E.~Fix and J.~L. Hodges.
\newblock {Discriminatory analysis, nonparametric discrimination: Consistency
  properties}.
\newblock {\em US Air Force School of Aviation Medicine}, Technical Report
  4(3):477+, January 1951.

\bibitem{cccg20afloresv}
Alejandro Flores{-}Velazco.
\newblock Social distancing is good for points too!
\newblock In {\em Proceedings of the 32st Canadian Conference on Computational
  Geometry, {CCCG} 2020, August 5-7, 2020, University of Saskatchewan,
  Saskatoon, Saskatchewan, Canada}, 2020.

\bibitem{DBLP:conf/cccg/Flores-VelazcoM19}
Alejandro Flores{-}Velazco and David~M. Mount.
\newblock Guarantees on nearest-neighbor condensation heuristics.
\newblock In {\em Proceedings of the 31st Canadian Conference on Computational
  Geometry, {CCCG} 2019, August 8-10, 2019, University of Alberta, Edmonton,
  Alberta, Canada}, 2019.

\bibitem{floresvelazco_et_al:LIPIcs:2020:12913}
Alejandro Flores-Velazco and David~M. Mount.
\newblock {Coresets for the Nearest-Neighbor Rule}.
\newblock In Fabrizio Grandoni, Grzegorz Herman, and Peter Sanders, editors,
  {\em 28th Annual European Symposium on Algorithms (ESA 2020)}, volume 173 of
  {\em Leibniz International Proceedings in Informatics (LIPIcs)}, pages
  47:1--47:19, Dagstuhl, Germany, 2020. Schloss Dagstuhl--Leibniz-Zentrum
  f{\"u}r Informatik.
\newblock URL: \url{https://drops.dagstuhl.de/opus/volltexte/2020/12913}, \href
  {http://dx.doi.org/10.4230/LIPIcs.ESA.2020.47}
  {\path{doi:10.4230/LIPIcs.ESA.2020.47}}.

\bibitem{DBLP:conf/esa/Flores-VelazcoM21}
Alejandro Flores{-}Velazco and David~M. Mount.
\newblock Boundary-sensitive approach for approximate nearest-neighbor
  classification.
\newblock In Petra Mutzel, Rasmus Pagh, and Grzegorz Herman, editors, {\em 29th
  Annual European Symposium on Algorithms, {ESA} 2021, September 6-8, 2021,
  Lisbon, Portugal (Virtual Conference)}, volume 204 of {\em LIPIcs}, pages
  44:1--44:15. Schloss Dagstuhl - Leibniz-Zentrum f{\"{u}}r Informatik, 2021.
\newblock URL: \url{https://doi.org/10.4230/LIPIcs.ESA.2021.44}, \href
  {http://dx.doi.org/10.4230/LIPIcs.ESA.2021.44}
  {\path{doi:10.4230/LIPIcs.ESA.2021.44}}.

\bibitem{avq_2020}
Ruiqi Guo, Philip Sun, Erik Lindgren, Quan Geng, David Simcha, Felix Chern, and
  Sanjiv Kumar.
\newblock Accelerating large-scale inference with anisotropic vector
  quantization.
\newblock In {\em International Conference on Machine Learning}, 2020.
\newblock URL: \url{https://arxiv.org/abs/1908.10396}.

\bibitem{guo2020accelerating}
Ruiqi Guo, Philip Sun, Erik Lindgren, Quan Geng, David Simcha, Felix Chern, and
  Sanjiv Kumar.
\newblock Accelerating large-scale inference with anisotropic vector
  quantization, 2020.
\newblock \href {http://arxiv.org/abs/1908.10396} {\path{arXiv:1908.10396}}.

\bibitem{avd-harpeled}
Sariel {Har-Peled}.
\newblock A replacement for {Voronoi} diagrams of near linear size.
\newblock In {\em Proc. 42nd Annu. IEEE Sympos. Found. Comput. Sci.}, pages
  94--103, 2001.

\bibitem{jankowski2004comparison}
Norbert Jankowski and Marek Grochowski.
\newblock Comparison of instances selection algorithms {I.} {A}lgorithms
  survey.
\newblock In {\em Artificial Intelligence and Soft Computing-ICAISC 2004},
  pages 598--603. Springer, 2004.

\bibitem{JDH17}
Jeff Johnson, Matthijs Douze, and Herv{\'e} J{\'e}gou.
\newblock Billion-scale similarity search with {GPUs}.
\newblock {\em arXiv preprint arXiv:1702.08734}, 2017.

\bibitem{johnson2017billionscale}
Jeff Johnson, Matthijs Douze, and Hervé Jégou.
\newblock Billion-scale similarity search with gpus, 2017.
\newblock \href {http://arxiv.org/abs/1702.08734} {\path{arXiv:1702.08734}}.

\bibitem{khodamoradi2018consistent}
Kamyar Khodamoradi, Ramesh Krishnamurti, and Bodhayan Roy.
\newblock Consistent subset problem with two labels.
\newblock In {\em Conference on Algorithms and Discrete Applied Mathematics},
  pages 131--142. Springer, 2018.

\bibitem{DBLP:journals/corr/abs-1905-01019}
Marc Khoury and Dylan Hadfield{-}Menell.
\newblock Adversarial training with {Voronoi} constraints.
\newblock {\em CoRR}, abs/1905.01019, 2019.
\newblock URL: \url{http://arxiv.org/abs/1905.01019}, \href
  {http://arxiv.org/abs/1905.01019} {\path{arXiv:1905.01019}}.

\bibitem{papernot2018deep}
Nicolas Papernot and Patrick McDaniel.
\newblock Deep $k$-nearest neighbors: Towards confident, interpretable and
  robust deep learning.
\newblock {\em arXiv preprint arXiv:1803.04765}, 2018.

\bibitem{peri2020deep}
Neehar Peri, Neal Gupta, W.~Ronny Huang, Liam Fowl, Chen Zhu, Soheil Feizi, Tom
  Goldstein, and John~P. Dickerson.
\newblock Deep $k$-nn defense against clean-label data poisoning attacks.
\newblock In {\em European Conference on Computer Vision}, pages 55--70.
  Springer, 2020.

\bibitem{DBLP:journals/corr/Schmidhuber14}
J{\"{u}}rgen Schmidhuber.
\newblock Deep learning in neural networks: An overview.
\newblock {\em CoRR}, abs/1404.7828, 2014.
\newblock URL: \url{http://arxiv.org/abs/1404.7828}, \href
  {http://arxiv.org/abs/1404.7828} {\path{arXiv:1404.7828}}.

\bibitem{sitawarin2019robustness}
Chawin Sitawarin and David Wagner.
\newblock On the robustness of deep $k$-nearest neighbors.
\newblock In {\em 2019 IEEE Security and Privacy Workshops (SPW)}, pages 1--7.
  IEEE, 2019.

\bibitem{stone1977consistent}
Charles~J. Stone.
\newblock Consistent nonparametric regression.
\newblock {\em The annals of statistics}, pages 595--620, 1977.

\bibitem{Wilfong:1991:NNP:109648.109673}
Gordon Wilfong.
\newblock Nearest neighbor problems.
\newblock In {\em Proceedings of the Seventh Annual Symposium on Computational
  Geometry}, SCG '91, pages 224--233, New York, NY, USA, 1991. ACM.
\newblock URL: \url{http://doi.acm.org/10.1145/109648.109673}, \href
  {http://dx.doi.org/10.1145/109648.109673} {\path{doi:10.1145/109648.109673}}.

\bibitem{Wilson:1997:IPT:645526.657143}
D.~Randall Wilson and Tony~R. Martinez.
\newblock Instance pruning techniques.
\newblock In {\em Proceedings of the Fourteenth International Conference on
  Machine Learning}, ICML '97, pages 403--411, San Francisco, CA, USA, 1997.
  Morgan Kaufmann Publishers Inc.
\newblock URL: \url{http://dl.acm.org/citation.cfm?id=645526.657143}.

\bibitem{Zukhba:2010:NPP:1921730.1921735}
A.~V. Zukhba.
\newblock {NP}-completeness of the problem of prototype selection in the
  nearest neighbor method.
\newblock {\em Pattern Recog. Image Anal.}, 20(4):484--494, December 2010.
\newblock URL: \url{http://dx.doi.org/10.1134/S1054661810040097}, \href
  {http://dx.doi.org/10.1134/S1054661810040097}
  {\path{doi:10.1134/S1054661810040097}}.

\end{thebibliography}

\end{document}